\documentclass[10pt, conference, compsocconf]{IEEEtran}
\pdfoutput=1

\usepackage{amssymb}

\usepackage{amsmath} \usepackage{amsthm} \usepackage{amsfonts}
\usepackage{graphicx} \usepackage{calc} \usepackage{subfigure}
\usepackage{booktabs} \usepackage{color}
\usepackage[algoruled,linesnumbered]{algorithm2e}
\SetKwInOut{Input}{input} \SetKwInOut{Output}{output}

\usepackage{tikz}
\usetikzlibrary{patterns,snakes}

\DeclareGraphicsExtensions{.pdf,.jpeg,.png}

\title{A packing problem approach to energy-aware load distribution in Clouds}

\author{\IEEEauthorblockN{Thomas Carli\IEEEauthorrefmark{1},
St\'e{}phane Henriot\IEEEauthorrefmark{2},
Johanne Cohen\IEEEauthorrefmark{3}, and
Joanna Tomasik\IEEEauthorrefmark{1}}
\IEEEauthorblockA{\IEEEauthorrefmark{1}Computer Science Dpt., E3S SUPELEC\\
91192 Gif sur Yvette, France\\ Email: \{Thomas.Carli, Joanna.Tomasik\}@supelec.fr}
\IEEEauthorblockA{\IEEEauthorrefmark{2}INSA Rouen\\
76800 St.-Etienne-du-Rouvray, France\\
Email: Stephane.Henriot@ens-cachan.fr}
\IEEEauthorblockA{\IEEEauthorrefmark{3}LRI, University Paris Sud\\ 91190 Gif sur Yvette, France\\
Email: Johanne.Cohen@lri.fr}
}

\begin{document}

\newtheorem{theorem}{Theorem}
\newtheorem{lemma}{Lemma}
\newtheorem{proposition}{Proposition}
\newtheorem{definition}{Definition}

\maketitle

\begin{abstract}
The Cloud Computing paradigm consists in providing customers with virtual services of the quality which meets customers' requirements. A cloud service operator is interested in using his infrastructure in the most efficient way while serving customers. The efficiency of infrastructure exploitation may be expressed, amongst others, by the electrical energy consumption of computing centers.

We propose to model the energy consumption of private Clouds, which provides virtual computation services, by a variant of the Bin Packing problem. This novel generalization is obtained by introducing such constraints as: variable bin size, cost of packing and the possibility of splitting items.

We analyze the packing problem generalization from a theoretical point of view. We advance on-line and off-line approximation algorithms to solve our problem to balance the load either on-the-fly or on the planning stage. In addition to the computation of the approximation factors of these two algorithms, we evaluate experimentally their performance.

The quality of the results is encouraging. This conclusion makes a packing approach a serious candidate to model energy-aware load balancing in Cloud Computing.

\end{abstract}

\section{Introduction \label{introduction}}

The Cloud Computing paradigm consists in providing customers with virtual services of the quality which meets customers' requirements. A cloud service operator is interested in using his infrastructure in the most efficient way while serving customers. Namely, he wishes to diminish the environmental impact of his activities by reducing the amount of energy consumed in his computing servers. Such an attitude allows him to lower his operational cost (electricity bill, carbon footprint tax, etc.) as well.

Three elements are crucial in the energy consumption on a cloud platform: computation (processing), storage, and network infrastructure~\cite{BalAyr+11,BelAba+12,BelBuy+11,BerlGel+10}.  We intend to study
different techniques to reduce the energy consumption regarding these three elements. We tempt to consolidate applications on servers to keep their utilization at hundred per cent. The consolidation problem was discussed in~\cite{SriKan+08} through an experimental approach based on the intuition as its authors did not propose any formal problem definition.

In this paper we address the challenge of the minimization of energy required for processing by means of proper
mathematical modeling and we propose algorithmic solutions to minimize the energy consumption on Cloud Computing platforms. We address here a private Cloud infrastructure which operates with knowledge of resource availability.

We study a theoretical problem adjacent to the minimization of energy required to execute computational tasks. Our working hypotheses are as follows:
\begin{enumerate}
\item any computional task is parallelizable, i.e. it may be executed on several servers; there is, however, a restriction on the number of servers on which a task can be launched,
\item available servers have different computing capacities,
\item the computation cost of a server in terms of its energy consumption is monotone, i.e. a unity of computation power is cheaper on a voluminous server than on a less capacious one.
\end{enumerate}

The assumption that all tasks are divisible may sound unrealistic as in practice some tasks cannot be split. We make it in order to formulate theoretical problems and analyze them. In the real world scenario one will rather cope with jobs which either cannot be cut at all or which can be cut one, twice, up to $D$ times. Such a situation corresponds to a problem which is ``somewhere between'' two extremal cases: no jobs can be split and all jobs can be split $D$ times. As the reader will notice going through this paper, the ``real life'' problem performance bounds can be deducted from those of the extremal problems.

The three assumptions above lead us to formulate a generalization of the Bin Packing problem~\cite{CofGar+97}, which we refer to as the Variable-Sized Bin Packing with Cost and Item Fragmentation Problem (VS-CIF-P). In the general case considered a cost of packing is monotone. This problem models a distribution of computational tasks on Cloud servers which ensures the lowest energy consumption. 
Its definition is given in Subsection~\ref{mainProblemDefinition}. We point out that the approach through packing problems to the energy-aware load distribution has not yet been proposed.

Confronted with numerous constraints of the VS-CIF-P we decided to start, however, by studying in Subsection~\ref{lessConstrantedProblemDefinition} a less constrained problem, without an explicit cost function, the Variable-Sized Bin Packing and Item Fragmentation Problem (VS-IF-P). This has not yet been studied either. This gradual approach allows us to deduce several theoretical properties of the VS-IF-P which can be then extended to the principal problem. 

In Section~\ref{algorithms} we propose customized algorithms to solve the VS-CIF-P.
Willing to treat users' demands in bulk, what corresponds to regular dispatching of collected jobs (for instance, hourly) we propose an off-line method (Subsection~\ref{greedyApproach}).
An on-line algorithm, dealing with demands on-the-fly is given in Subsection~\ref{algorithmWithPerformanceBounds}. This treatment allows one to launch priority jobs which have to be processed upon their arrivals.
Expecting an important practical potential of the VS-CIF-P we also furnish results concerning the theoretical performance bounds of the algorithms we elaborated.

Despite the fact that the problem is approximate with a constant factor, we go further with the performance evaluation of the algorithms we come up with.
The empirical performance evaluation is discussed in Section~\ref{numericalResults}. 

 
The list of our contributions given above also partially constitutes the description of the paper's organization. We complete this description by saying that in Section~\ref{relatedWorks} we present a survey of related works concerning definitions of the family of bin-packing problems together with their known approximation factors. We also give there an outline of algorithmic approaches used to solve packing problems. Our special attention is put on those which inspired us in our study. We point out that the notation used in the article is also introduced in that section while carrying out our survey. 

After giving our contributions in the order announced above we draw conclusions and give directions of our further work. 

\section{Related Works \label{relatedWorks}}

Let $L$ be a list of $n$ items numbered from $1$ to $n$, $L=(s_1, s_2, \ldots, s_n)$, where $s_i$ indicates an item size. Let us also assume for a moment that for all $i=1,2,\ldots, n$ $s_i\in [0,1]$. The classical Bin Packing Problem (BPP) consists in grouping the items of $L$ into $k$ disjoint subsets, called bins, $B_1, B_2, \ldots, B_k$, $\bigcup_{l=1}^m B_l = L$, such that for any $j$, $j=1,2,\ldots,k$, $\sum_{l\in B_j}s_l \leq 1$. The question 'Can I pack all items of $L$ into $K$, $K\leq k$, bins?' defines the BPP in the decision form. Put differently, we ask whether a \emph{packing\/} $B_1, B_2, \ldots, B_k$ for which $k$ is less or equal to a given value $K$ exists. The corresponding optimization problem aims to find the minimal $k$. Due to its numerous practical applications the BPP, which is NP-hard, was studied exhaustively.

The current basic on-line approaches, Next Fit (NF) and First Fit (FF) give satisfactory results. The asymptotic approximation factor for any on-line algorithm cannot be less 
than $1.54$~\cite{CofGar+97}. A widely used off-line approach consists in sorting items in decreasing order of their size before packing them. The tight bound for First Fit Decreasing (FFD) is given in~\cite{Dos07}.

\subsection{Variable-Sized Bin Packing with Cost\label{VSBPCP}}

In the initial problem the capacity of all bins is unitary. The problem may thus be modified by admitting different bin capacities.
A bin can have any of $m$ possible capacities $b_j$, $B= (b_1, b_2, \ldots, b_m)$. In other words, we have $m$ bin classes.

If any bin is as good as the others, putting items inside a solution to this problem is trivial, as one will always be  interested in using the largest bin. We thus suppose that bin utilization induces a certain cost associated to this bin. This assumption leads to the Variable Sized Bin Packing with Cost Problem (VSBPCP). The reader might already observe that in the BPP the cost of packing is always the same regardless of the bins chosen. This fact explains that the new problem is NP-hard~\cite{KanPar03}. Intuitively, one can consider that the cost of packing varies in function of a bin capacity. From this point of view we are no longer interested in minimizing the number of bins used but in minimizing the global packing cost as a voluminous bin which remains 'almost empty' may be more expensive in use than several little bins 'almost totally' full. Solving the VSBPCP we have at our disposal $m$ classes of bins and the infinite number of bins of any class available. 

In the simplest case, the cost is a linear function of a capacity.
Packing into bin $i$ costs $c_i$, $C= (c_1, c_2, \ldots, c_m)$ and we have as many costs as bin classes available.
Similarly to the notation introduced above, we denote a cost of a bin $B_j$ taking part in a packing as $\mbox{cost}(B_j) = c_l$ (a bin in position $j$ in a packing costs $c_l$).
The goal is thus to find a packing $B_1, B_2, \ldots, B_k$ such that $\sum_{l\in B_j}s_l \leq \mbox{capacity}(B_j)$ for which the overall cost, $\sum_{l=1}^k \mbox{cost}(B_l)$, is minimal.

Without loss of generality one may assume that the cost of using bin $i$ is equal to its capacity, $c_i=b_i$, $i=1,2,\ldots, m$.
For a packing we thus have: $\mbox{cost}(B_j) = \mbox{capacity}(B_j)$.  

A monotone cost function signifies that a unity in a bigger bin $i$ is not more expensive than a unity in a smaller one, bin $j$: $\frac{c_i}{b_i}\leq \frac{c_j}{b_j}$, $i,j=1,2,\ldots,m$, $b_i>b_j$ and the cost of a smaller bin is not greater than the cost of a bigger bin,  $c_j\leq c_i$. A linear cost function is a special case of a monotone one.

\subsubsection{Monotone cost, off-line approach \label{IFFD}}

Staying in the context of a monotone cost, our attention was attracted by the off-line algorithms from~\cite{KanPar03}. Their main idea consists in applying an iterative approach to a well-known algorithm, for example, FFD which leads to IFFD.

In a nutshell, at the beginning IFFD performs the classical FFD with identical bins of the greatest capacity, $\max_{b_i}(B)$. The packing obtained is next modified by trying to move (again with FFD) all the items from the last bin into the next biggest bin. The repacking procedure continues by transferring  items entirely from the bin of capacity $b_j$, which was the last one filled up, into a bin of size $b_i$, $b_i<b_j$ and there is not any $l$ such that $b_i<b_l<b_j$. It stops when any further repacking becomes impossible. Their authors showed that the solutions are approximated with $1.5$.   

\subsubsection{Linear cost, on-line approach \label{FFf}}

We pay special attention to the on-line algorithm to solve the VSBPCP described in~\cite{KinLan89}. This algorithm deals with a linear cost.
Its authors proposed an approach 'in between' the First Fit, using Largest possible bin (FFL)\footnote{A largest possible bin is here a unit-capacity bin.} and the First Fit, using Smallest possible bin (FFS) trying to take advantage of both methods with regard to the size of the item to be packed. Their idea is to determine whether an item to be packed occupies a lot of space or not. The decision is taken upon a fill factor $f$, $f\in [0.5,1]$. Their algorithm, called FFf, operates in the following way. If item $i$ is small (i.e. $s_i\leq 0.5$), it will be inserted into the first bin in which it enters or into a new bin of unitary capacity when it does not fit into any opened bin (FFL). Otherwise, it will be inserted into the first opened bin into which it enters. If the use of an opened bin is impossible, it will be packed into the smallest bin among those whose capacity is between $s_i$ and $\frac{s_i}{f}$, if it fits inside, or into a new unit-capacity bin if it does not (FFS). The authors of FFf proved 
that the result it furnishes is approximated by $1.5 + \frac{f}{2}$.

\subsection{Bin Packing with Item Fragmentation \label{BPIFP}}

In another variant of the classical BPP one is allowed to fragment items while the identical bin size and bin cost remain preserved, the Bin Packing with Item Fragmentation Problem (BPIFP).
Item cutting may reduce the number of bins required. On the other hand, if the item fragmentation is not for free, it may increase the overall cost of packing.
In~\cite{MenRom01,NaaRom02,ShaTam+06}, for instance, its authors investigated two possible expenses: the item size growth which results from segmentation and the global limit on the number of items cut.
We point out that one may also consider a limit on the number of items permitted to be packed into a bin. A variant of the BPP fixing such a limit was introduced and studied in~\cite{KraShe+75,KraShe+77}.
It models task scheduling in multiprogramming systems and is known as the Bin Packing with Cardinality Constraints Problem (BPCCP). 

From our particular perspective, founded upon the virtualization of computing services in Clouds, we opt to restrain the number of fragments into which an item can be cut.
The maximal number of cuts for any item, which models a computation task, is limited to $D$.
As certain items from list $L$ should be fragmented before packing, we do not cope with items $i$ but with their fragments whose sizes are noted as $s_{i_d}$, where $i$ indicates an original item $i$ from $L$ and $d$ enumerates fragments of item $i$. Let $D_i$ be a number of cuts of item $i$ made. Obviously, $D_i=0$ signifies that item $i$ has not been fragmented at all. Moreover, for any $i$ we have $D_i\leq D$ and $\sum_{l=1}^{D_i +1} s_{i_l} = s_i $. Thus the solution to the BPIFP with limit $D$ consists in finding an appropriate fragmentation first, which results in a new list of sizes of items to be packed
\begin{equation}
\begin{array}{lll}
L_D & = & \left((s_{1_1}, s_{1_2}, \ldots, s_{1_{D_1 +1}}), (s_{2_1}, s_{2_2}, \ldots, s_{2_{D_2 +1}}), \ldots, \right. \\
& & \left. \ldots, (s_{n_1}, s_{n_2}, \ldots, s_{n_{D_n +1}})\right)
\end{array}
\label{LD}
\end{equation}
(actually, this a list of lists). The number of items to be packed is now $\sum_{i=1}^{n}(D_i + 1)$. Next, the BPP is to be solved with $L_D$ as input data.

\subsection{Survey conclusions \label{surveyConclusions}}
To the best of our knowledge, neither the problem being in the center of our interest, the Variable-Sized Bin Packing with Cost and Item Fragmentation Problem (VS-CIF-P), which we propose to model an energy-aware load distribution nor the less constrained one, the Variable-Sized Bin Packing and Item Fragmentation Problem (VS-IF-P), have been studied yet. 

\section{Problem Definitions and Analysis \label{problemDefinitions}}

As announced above, we start by treating the auxiliary problem. It will be generalized after its analysis.

\subsection{Auxiliary Problem \label{lessConstrantedProblemDefinition}} 

We suppress the explicit cost function in the general problem. By doing this we expect to be able to find an optimal solution to the auxiliary VS-IF-P with polynomial complexity for certain particular cases. We recall to the reader here that we limit the number of cuts of any individual item.

\begin{definition}
\label{defVSIFP}
{\bf Variable-Sized Bin Packing with Item Fragmentation Problem (VS-IF-P)}\\
{\bf Input:}
	\begin{itemize}
	\item $n$ items to be packed,
	\item sizes of items to be packed $L=(s_1, s_2, \ldots, s_n)$, $s_i \in \mathbb{N^+}$, $i=1,2,\ldots n$,
	\item capacities of bins available $B= (b_1, b_2, \ldots, b_m)$, $ b_j \in \mathbb{N^+}$, $j=1,2,\ldots m$,
	\item a constant $D$ which limits the number of splits authorized for each item, $D \in \mathbb{N^+}$,
	\item a constant $k$, $k \in \mathbb{N^+}$ which signifies the number of bins used.
	\end{itemize}		
\noindent {\bf Question:}\\
	Is it possible to find a packing $B_1, B_2, \ldots, B_K$ of items $L$ whose fragment sizes are in $L_D$ defined as in Eq.~\eqref{LD} such that $K\leq k$?
\end{definition}	

In the analysis of the VS-IF-P we appeal to a variant of the BPP coming from the memory allocation modeling~\cite{ChuGra+06,EpsSte11}. Its particularity consists in having a limit on the number of items in any bin. This limit holds without regard to whether an item inserted is 'an original one' or results from the item fragmentation itself. Thus this problem may be considered as a variant of the BPCCP (see Subsection~\ref{BPIFP}) with item fragmentation. For our purposes we call it the Memory Allocation with Cuts Problem (MACP) and we give below its formal definition using our notation. We believe that this formal presentation allows the reader to discover a palpable duality existing between the VS-IF-P and the MACP: in the first one there is a limit on the cut number of any item, in the latter we have a limit on the number of 'cuts' in any bin.

As the MACP admits the item fragmentation, the objects which it packs are picked from the following list:
\begin{equation}
\begin{array}{lll}
L'_{D'} & = & \left((s'_{1_1}, s'_{1_2}, \ldots, s'_{1_{d_1 +1}}), (s'_{2_1}, s'_{2_2}, \ldots, s'_{2_{d_2 +1}}), \ldots, \right.\\
 & &\left. \ldots, (s'_{n_1}, s'_{n_2}, \ldots, s'_{n_{d_n +1}})\right),
\end{array}
\label{LdMACP}
\end{equation}
where $\sum_{l=1}^{d_i + 1}s'_{i_l} = s'_i$, $i=1,2,\ldots, n$. The number of splits of item $i$, $d_i$, is not \emph{a priori\/} limited but one cannot cut items at will. The values of $d_i$ will be determined later on by the constraint restricting the number of pieces in a bin, $D'$. 

\begin{definition}
\label{defMACP}
{\bf Memory Allocation with Cuts Problem (MACP)}\\
{\bf Input:}
	\begin{itemize}
	\item $n'$ items to be packed,
	\item sizes of items to be packed $L'=(s'_1, s'_2, \ldots, s'_n)$, $s'_i \in \mathbb{N^+}$, $i=1,2,\ldots n'$,
	\item a capacity $b'$ of each bin,
	\item a constant $D'$ which limits the number of items authorized inside each bin (no more than $D'+1$ pieces inside a bin), $D' \in \mathbb{N^+}$,
      \item a constant $k'$, $k' \in \mathbb{N^+}$ which signifies the number of bins used.
	\end{itemize}		
\noindent {\bf Question:}\\
	Is it possible to find a packing $B'_1, B'_2, \ldots, B'_{K'}$ of elements whose sizes are in $L'_{D'}$ defined as in Eq.~\eqref{LdMACP} such that for each $l$, $l=1,2,\ldots K'$, $|B'_l|\leq D'$ and $K'\leq k'$?
\end{definition}	

\begin{theorem}
The VS-IF-P is NP-complete in the strong sense.
\label{VSIFPtheoremNP}
\end{theorem}

\begin{proof}
It is easy to see that the VS-IF-P is in NP. In this proof, as in Def.~\ref{defMACP}, we consequently use  a prime symbol when referring to an instance of the MACP. 
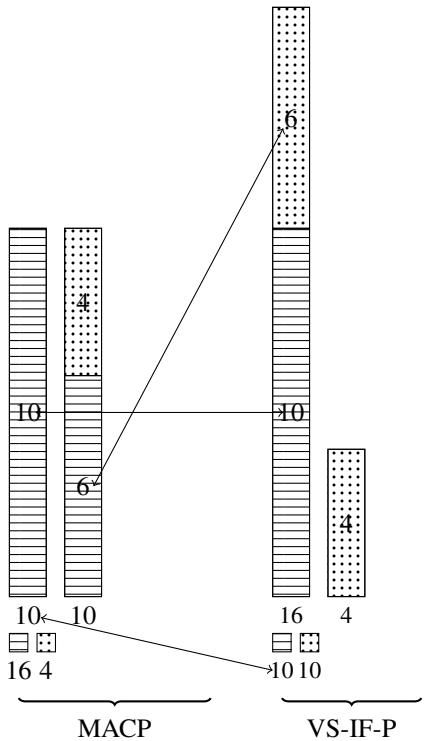
\begin{figure}[!h]
 	\centering
\begin{tikzpicture}[scale=0.7]
\draw (0.0, 0) rectangle (0.7, 7.0);
\draw (1.0499999999999998, 0) rectangle (1.75, 7.0);

\node at (0.35, -0.35) {10};
\node at (1.4, -0.35) {10};

\draw [pattern=horizontal lines] (0.0, 0.0) rectangle (0.7, 7.0);
\draw [pattern=horizontal lines] (1.0499999999999998, 0.0) rectangle (1.75, 4.199999999999999);
\draw [pattern=dots] (1.0499999999999998, 4.199999999999999) rectangle (1.75, 7.0);

\node at (0.35, 3.5) {10};
\node at (1.4, 2.0999999999999996) {6};
\node at (1.4, 5.6) {4};

\draw [pattern=horizontal lines] (0.0, -1.0499999999999998) rectangle (0.35, -0.7);
\draw [pattern=dots] (0.5249999999999999, -1.0499999999999998) rectangle (0.875, -0.7);

\node at (0.175, -1.4) {16};
\node at (0.7, -1.4) {4};

\node (fix) at (0,-1.5) {};
\node (fix2) at (4,-1.5) {};
\draw [
    thick,
    decoration={
        brace,
        mirror,
        raise=0.3cm
    },
    decorate
] (fix) -- (fix2) ;
\node[draw=none] at (2,-2.5) {MACP};

\draw[arrows=<->](0.6,-0.4)--(5,-1.4);
\draw[arrows=<->](0.5,3.5)--(5.2,3.5);
\draw[arrows=<->](1.6,2.1)--(5.2,8.9);

\draw (5, 0) rectangle (5.7, 11.2);
\draw (6.0499999999999998, 0) rectangle (6.75, 2.8);

\node at (5.35, -0.35) {\small 16};
\node at (6.4, -0.35) {\small 4};

\draw [pattern=horizontal lines] (5, 0.0) rectangle (5.7, 7.0);
\draw [pattern=dots] (5.0, 7.0) rectangle (5.7, 11.2);
\draw [pattern=dots] (6.0499999999999998, 0.0) rectangle (6.75, 2.8);

\node at (5.35, 3.5) {10};
\node at (5.35, 9.1) {6};
\node at (6.4, 1.4) {4};

\draw [pattern=horizontal lines] (5.0, -1.0499999999999998) rectangle (5.35, -0.7);
\draw [pattern=dots] (5.5249999999999999, -1.0499999999999998) rectangle (5.875, -0.7);

\node at (5.175, -1.4) {\small 10};
\node at (5.7, -1.4) {\small 10};

\node (fix) at (5,-1.5) {};
\node (fix2) at (8,-1.5) {};
\draw [
    thick,
    decoration={
        brace,
        mirror,
        raise=0.3cm
    },
    decorate
] (fix) -- (fix2) ;
\node[draw=none] at (6.5,-2.5) {VS-IF-P};

\end{tikzpicture}
	\caption{MACP/VS-IF-P instance transformation}
	\label{instancesProofVSIFPtheoremNP}
 \end{figure} 
 
First, we demonstrate that the VS-IF-P is NP-complete in the strong sense by reducing the MACP to it as the NP-completeness in the strong sense of the MACP was shown in~\cite{EpsSte11}.
For both instances, $I$ and $I'$, we put $D=D'$.

Bins of the MACP become $k'$ items to be packed in the VS-IF-P, $L=(b',b',\ldots, b')$, $n=k'$. Items of the MACP are transformed into variable-sized bins, $L'=B$. Finally, we require that all available bins of the VS-IF-P are used: $|L'|= |B| =k$. Obviously, this transformation, also illustrated  in Figure~\ref{instancesProofVSIFPtheoremNP}, can be performed in polynomial time.

We focus our attention on a particular case $k'b'=\sum_{b_i\in L'}b_i$ in which there is no empty space left in bins forming a solution.
It is evident that in this situation the packing of items into bins corresponds to 'inserting' bins into items. 
If the verification gives a positive answer for one instance, it will give a positive answer for another, too. An illustrative example, $D=d'=1$, is given in Figure~\ref{exampleProofVSIFPtheoremNP}.
The overall items' mass is $40$. Thus we need  four bins of capacity $10$ to pack $L'=(16,15,9)$ for the instance $I'$ of the MACP depicted on the left of Figure~\ref{exampleProofVSIFPtheoremNP}.
The instance $I$ of the VS-IF-P (on the right of Figure~\ref{exampleProofVSIFPtheoremNP}) is composed of four items of size $10$ and $B=(16,15,9)$.  
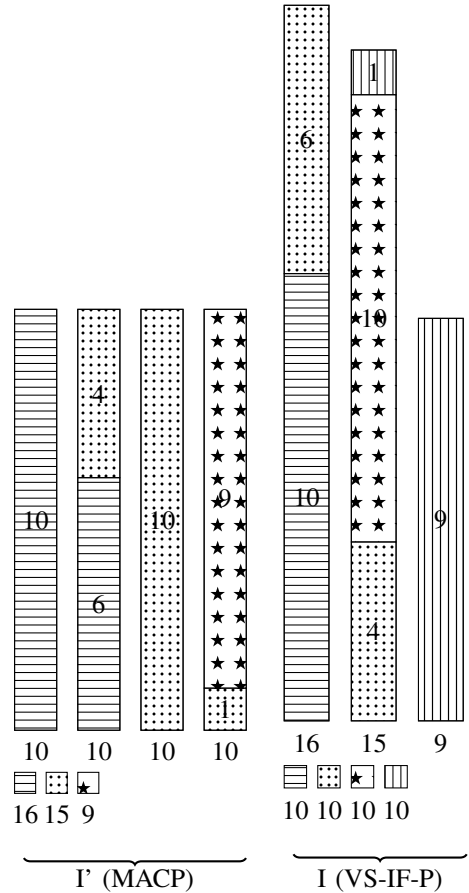
\begin{figure}[!h]
 	\centering
        \begin{tikzpicture}[scale=0.8]
\draw (0.0, 0) rectangle (0.7, 7.0);
\draw (1.0499999999999998, 0) rectangle (1.75, 7.0);
\draw (2.0999999999999996, 0) rectangle (2.8, 7.0);
\draw (3.15, 0) rectangle (3.8499999999999996, 7.0);
\node at (0.35, -0.35) {10};
\node at (1.4, -0.35) {10};
\node at (2.4499999999999997, -0.35) {10};
\node at (3.5, -0.35) {10};
\draw [pattern=horizontal lines] (0.0, 0.0) rectangle (0.7, 7.0);
\draw [pattern=horizontal lines] (1.0499999999999998, 0.0) rectangle (1.75, 4.199999999999999);
\draw [pattern=dots] (1.0499999999999998, 4.199999999999999) rectangle (1.75, 7.0);
\draw [pattern=dots] (2.0999999999999996, 0.0) rectangle (2.8, 7.0);
\draw [pattern=dots] (3.15, 0.0) rectangle (3.8499999999999996, 0.7);
\draw [pattern=fivepointed stars] (3.15, 0.7) rectangle (3.8499999999999996, 7.0);
\node at (0.35, 3.5) {10};
\node at (1.4, 2.0999999999999996) {6};
\node at (1.4, 5.6) {4};
\node at (2.4499999999999997, 3.5) {10};
\node at (3.5, 0.35) {1};
\node at (3.5, 3.8499999999999996) {9};
\draw [pattern=horizontal lines] (0.0, -1.0499999999999998) rectangle (0.35, -0.7);
\draw [pattern=dots] (0.5249999999999999, -1.0499999999999998) rectangle (0.875, -0.7);
\draw [pattern=fivepointed stars] (1.0499999999999998, -1.0499999999999998) rectangle (1.4, -0.7);
\node at (0.175, -1.4) {16};
\node at (0.7, -1.4) {15};
\node at (1.2249999999999999, -1.4) {9};
\node (fix) at (0,-1.5) {};
\node (fix2) at (4,-1.5) {};
\draw [
    thick,
    decoration={
        brace,
        mirror,
        raise=0.5cm
    },
    decorate
] (fix) -- (fix2) ;
\node[draw=none] at (2,-2.5) {I' (MACP)};

\end{tikzpicture}
\begin{tikzpicture}[scale=0.85]
\draw (0.0, 0) rectangle (0.7, 11.2);
\draw (1.0499999999999998, 0) rectangle (1.75, 10.5);
\draw (2.0999999999999996, 0) rectangle (2.8, 6.3);
\node at (0.35, -0.35) {16};
\node at (1.4, -0.35) {15};
\node at (2.4499999999999997, -0.35) {9};
\draw [pattern=horizontal lines] (0.0, 0.0) rectangle (0.7, 7.0);
\draw [pattern=dots] (0.0, 7.0) rectangle (0.7, 11.2);
\draw [pattern=dots] (1.0499999999999998, 0.0) rectangle (1.75, 2.8);
\draw [pattern=fivepointed stars] (1.0499999999999998, 2.8) rectangle (1.75, 9.799999999999999);
\draw [pattern=vertical lines] (1.0499999999999998, 9.799999999999999) rectangle (1.75, 10.5);
\draw [pattern=vertical lines] (2.0999999999999996, 0.0) rectangle (2.8, 6.3);
\node at (0.35, 3.5) {10};
\node at (0.35, 9.1) {6};
\node at (1.4, 1.4) {4};
\node at (1.4, 6.3) {10};
\node at (1.4, 10.149999999999999) {1};
\node at (2.4499999999999997, 3.15) {9};
\draw [pattern=horizontal lines] (0.0, -1.0499999999999998) rectangle (0.35, -0.7);
\draw [pattern=dots] (0.5249999999999999, -1.0499999999999998) rectangle (0.875, -0.7);
\draw [pattern=fivepointed stars] (1.0499999999999998, -1.0499999999999998) rectangle (1.4, -0.7);
\draw [pattern=vertical lines] (1.575, -1.0499999999999998) rectangle (1.9249999999999998, -0.7);
\node at (0.175, -1.4) {10};
\node at (0.7, -1.4) {10};
\node at (1.2249999999999999, -1.4) {10};
\node at (1.75, -1.4) {10};
\node (fix) at (0,-1.5) {};
\node (fix2) at (3,-1.5) {};
\draw [
    thick,
    decoration={
        brace,
        mirror,
        raise=0.5cm
    },
    decorate
] (fix) -- (fix2) ;
\node[draw=none] at (1.5,-2.5) {I (VS-IF-P)};
\end{tikzpicture}
	\caption{Example of solutions to MACP and VS-IF-P instances ($I'$ and $I$, respectively). Instance $I'$ is composed of three items whose sizes are $16$, $15$, $9$ and bins of capacity $10$. Instance $I$ has four items of size $10$ and three bin classes of capacity: $16$, $15$, and $9$.}
	\label{exampleProofVSIFPtheoremNP}
 \end{figure}

Second, we estimate the computational effort required to obtain a positive response to the question whether 'a candidate to be a solution' is a solution.
In order to do this we determine a size $N_i$ of the VS-IF-P instances and a size $N_s$ of solutions to it.
The greatest elements of both lists determine $N_i$, which is thus in $\mathcal{O} (\log k + (m+n)\log \max(\max_{s_i}(L), \max_{b_i}(B)))$.
Solutions are made up of bins and 'quantities' of items, possibly split, selected in $L$. Similarly, we take into account the most voluminous elements of the lists, which leads us to $N_s$ not greater than $nm(\log m + \log \max(\max_{s_i}(L), \max_{b_i}(B)))$, being in $\mathcal{O}( N_i^2)$ (a polynomial verification time).
\end{proof}

In order to analyze the feasibility of solutions to the VS-IF-P we assume for a moment that fragments resulting from the item splitting are equal in size. We have to be assured that any fragment of the greatest item can be inserted entirely into the highest capacity bin:
\begin{equation} \frac{\max_{s_l} (L)}{D+1}\leq \max_{b_i} (B).
\label{cutEqualParts} \end{equation}
If we can pack items of an instance of the VS-IF-P with such a fragmentation, we can do the same for the VSBPP instance, whose items to be packed are simply those of the VS-IF-P split. This reasoning allows us to adapt numerous algorithms existing for the VSBPP to solve the VS-IF-P by incorporating cutting. Eq.~\eqref{cutEqualParts} guarantees the existence of a solution  even in cases when items are split into 'almost equal' pieces because a single 'over-sized' fragment will be not greater than $\frac{\max_{s_l} (L)}{D+1}$. We propose to admit cutting in FF (Cut and FF, CFF). Despite sorting the bin classes in decreasing order of their capacity, the on-line principle is preserved as items are not reordered before their cut and insertion. An illustration of a CFF execution with 'imperfect cuts' is presented in Figure~\ref{exampleCFF}.
\begin{figure}[h]
 	\centering
        \begin{tikzpicture}
\draw (0.0, 0) rectangle (0.75, 7.5);
\draw (1.125, 0) rectangle (1.875, 7.5);
\draw (2.25, 0) rectangle (3.0, 7.5);
\draw (3.375, 0) rectangle (4.125, 7.5);
\draw (4.5, 0) rectangle (5.25, 7.5);
\draw (5.625, 0) rectangle (6.375, 7.5);
\node at (0.375, -0.375) {10};
\node at (1.5, -0.375) {10};
\node at (2.625, -0.375) {10};
\node at (3.75, -0.375) {10};
\node at (4.875, -0.375) {10};
\node at (6.0, -0.375) {10};
\draw [pattern=horizontal lines] (0.0, 0.0) rectangle (0.75, 5.25);
\draw [pattern=horizontal lines] (3.375, 0.0) rectangle (4.125, 4.5);
\draw [pattern=dots] (1.125, 0.0) rectangle (1.875, 5.25);
\draw [pattern=dots] (4.5, 0.0) rectangle (5.25, 4.5);
\draw [pattern=fivepointed stars] (2.25, 0.0) rectangle (3.0, 5.25);
\draw [pattern=fivepointed stars] (5.625, 0.0) rectangle (6.375, 4.5);
\node at (0.375, 2.625) {7};
\node at (3.75, 2.25) {6};
\node at (1.5, 2.625) {7};
\node at (4.875, 2.25) {6};
\node at (2.625, 2.625) {7};
\node at (6.0, 2.25) {6};
\draw [pattern=horizontal lines] (0.0, -1.125) rectangle (0.375, -0.75);
\draw [pattern=dots] (0.5625, -1.125) rectangle (0.9375, -0.75);
\draw [pattern=fivepointed stars] (1.125, -1.125) rectangle (1.5, -0.75);
\node at (0.1875, -1.5) {13};
\node at (0.75, -1.5) {13};
\node at (1.3125, -1.5) {13};
\end{tikzpicture}
	\caption{Example CFF solving an VS-IF-P instance; observe that an imperfect cut do not compromise the solution feasibility}
	\label{exampleCFF}
 \end{figure}
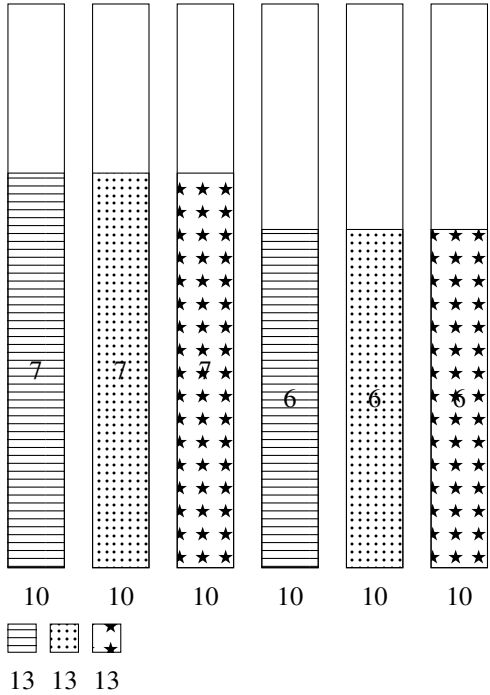
The discussion above, which exhibits a relationship between instances of the VS-IF-P and the VSBPP allows us to apply FF (or NF). In solutions obtained with these algorithms one bin at most is less than half full. 

On the other hand, if bins are of capacity $2s$, where $s$ is a natural number, items are of size $2(s+1)$ and any item can be cut no more than once, $D=1$, we cannot expect a better approximation factor than $2$. We visualize this example by imagining in~Figure~\ref{exampleCFF} the item size equal to $12$ and keeping bin capacity equal to $10$.  This observation leads us to formulate:
\begin{theorem} 
A solution to the VS-IF-P for any $D$ with CFF is tightly bounded by $2$.
\label{theoremApproxCFF} 
\end{theorem}   

\subsection{Main Problem \label{mainProblemDefinition}}

As stated above, the problem which models the distribution of computation tasks within a private Cloud infrastructure is a bin packing in which bins are of different sizes and tasks can be split over several servers.
We assume that all numerical data (item sizes, bin capacities, costs) are natural numbers. We also assume that a task is parallelizable (the discussion of our hypothesis can be found in Section~\ref{introduction}).
We allow an item to be cut into no more than $D+1$ pieces, i.e. any item may be split at most $D$ times. Indeed, if any number of cuts was admitted, we might cut all tasks into unitary pieces and end up with a trivial packing of $\sum_{i=1}^n s_i$ unitary objects being able to fill up any used bin entirely.  

As the reader might have already notice while passing through Sections~\ref{introduction} and~\ref{relatedWorks}, our problem, the Variable-Sized Bin Packing with Cost and Item Fragmentation Problem (VS-CIF-P), puts together the three problems announced above.  To be more precise, we 'mix up' the VSBPCP and BPIFP, the latter with the constraints which have just been discussed. Unless stated differently, the cost function is monotone.
  
\begin{definition}
\label{defVSLCIFP}
{\bf Variable-Sized Bin Packing with Cost and Item Fragmentation Problem (VS-CIF-P)}\\
{\bf Input:}
	\begin{itemize}
	\item $n$ items to be packed,
	\item sizes of items to be packed $L=(s_1, s_2, \ldots, s_n)$,$s_i \in \mathbb{N^+}$, $i=1,2,\ldots n$,
	\item capacities of bins available $B= (b_1, b_2, \ldots, b_m)$, $ b_j \in \mathbb{N^+}$, $j=1,2,\ldots m$,
	\item costs of using of bins available $C= (c_1, c_2, \ldots, c_m)$, $ c_j \in \mathbb{N^+}$, $j=1,2,\ldots m$,
	\item a constant $D$ which limits the number of splits authorized for each item, $D \in \mathbb{N^+}$,
	\item a constant $e$, $e \in \mathbb{N^+}$ which signifies the cost limit of a packing.
	\end{itemize}		
\noindent {\bf Question:}\\
	Is it possible to find a packing $B_1, B_2, \ldots, B_k$ of items $L$  whose fragment sizes are in $L_D$ defined as in Eq.~\eqref{LD} such that $\sum_{l=1}^k \mbox{cost}(B_l)\leq e$?
\end{definition}

\section{Algorithms \label{algorithms}}

Before introducing our methods to solve the VS-CIF-P (Def.~\ref{defVSLCIFP}) we will discuss the solution to the auxiliary VS-IF-P (Def.~\ref{defVSIFP}), in order to select the algorithmic approaches the best adapted to treat our principal problem. Indeed, as we presumed (Subsection~\ref{lessConstrantedProblemDefinition}), the VS-IF-P can be solved exactly and in polynomial time under a certain hypothesis.

\subsection{Next Fit with Cuts (NFC) for the auxiliary problem \label{NFC}}
Let us hypothesize that we are dealing with the instances for which we can always find a bin to pack any entire item (without cutting it). This hypothesis may be expressed by:
\begin{equation}
\max_{s_i} (L) \leq \max_{b_i} (B).
\label{StrongH}
\end{equation} 
If this hypothesis is satisfied, we propose a variant of NF, Next Fit with Cuts (NFC) as an algorithmic solution. A similar approach was used for another purpose under the name of NF$_{\mbox{\scriptsize f}}$ in~\cite{MenRom01,NaaRom02}. First, the capacities of bin classes are sorted in decreasing order.
 Next, for each item, if there is some room in a current bin, we pack the item inside, cutting it if necessary, and inserting the second fragment of the item into the next bin. We observe that with Hypothesis~\eqref{StrongH} valid, NFC fragments any item at most once. Moreover, when this hypothesis is satisfied, the packing problem with variable-sized bins and item fragmentation is in $P$.

Indeed, the verification whether Hypothesis~\eqref{StrongH} holds or not can be performed in $\mathcal{O} (m)$ or $\mathcal{O}(n)$ depending upon the relationship existing between $m$ and $n$. An execution of NFC requires $\mathcal{O}(m \log m + \max (n,m))$ operations. Lastly, we observe that the bins used are all totally filled up and they are of the greatest available capacities, which proves the algorithm's optimality.

This discussion leads us to the conclusion:
\begin{theorem}
NFC is optimal and polynomial to solve the VS-IF-P when Hypothesis~\eqref{StrongH} holds.
\label{theoremNFC}
\end{theorem}

   \subsection{Off-line Approach to the Main Problem with Monotone Cost \label{greedyApproach}}

The approach presented here is based upon IFFD (Paragraph~\ref{IFFD}) combined with item cutting. For this reason we refer to it as CIFFD. 
As before, we assume that bin classes are sorted in decreasing capacity order, $b_1>b_2> \cdots b_m$. 

The initial idea of our algorithm to solve the VS-CIF-P (Algorithm~\ref{greedyForVSLCIFP}) consists in dividing items of $L$ into two categories: those items $i$ which may be possibly packed without fragmentation and those which undoubtedly may not. Such an approach makes our algorithm off-line. We propose to reason here upon item sizes $s_i$, not upon their indices $i$. This mental operation enables us to avoid tedious renumbering of items to be packed, which might deteriorate the text limpidity. At the same time, it does not introduce any ambiguity.

 We formally note  the two categories as $T_1$ and $T^+_1$, respectively: $T_1 \cup T^+_1 = \{s_1,s_2, \ldots, s_n\}$, for all $s_i\in T_1$ we have $s_i\leq \max_{b_i}(B)=b_1$, and for all $s_i\in T^+_1$ we have $ s_i > \max_{b_i}(B)=b_1$. Items from $T^+_1$ are cut naturally up to $D$ times in order to completely fill up a bin of capacity $b_1$, the remaining fragment whose size is inferior to $b_1$ is stored in $T^-_1$ (lines~\ref{feasibilityAlgoOFFLoopBegin}--\ref{feasibilityAlgoOFFLoopEnd} of Algorithm~\ref{greedyForVSLCIFP}). This loop also allows us to detect the instance infeasibility, i.e. the number of cuts allowed $D$ is too small to insert an item fragmented into the largest bins.

The items whose sizes are in $T_1 \cup T^-_1$ are then packed according an appropriate algorithm to solve the BPP as we use momentarily bins of identical capacity $b_1$. We have thus a solution which we try to improve iteratively, taking bins in decreasing order of their capacity (for a bin of capacity $b_j$ items are divided into $T_j$ and $T_j^+$), by consecutive repacking of the contents of the less efficiently used bin into a smaller empty one, if possible (as IFFD described in Paragraph~\ref{IFFD} does). As our algorithm can fragment an item to fill up a bin, its iterative descent may stop when the number of cuts allowed has been reached. A solution which offers the lowest cost among the obtained ones is returned. Finally, an attempt is made to squeeze this solution more (lines \ref{repackingAlgoOFFLoopBegin}--\ref{repackingAlgoOFFLoopEnd} of Algorithm~\ref{greedyForVSLCIFP}).   
  
\begin{algorithm}[!h]
\KwData{VS-CIF-P data as in Def.~\ref{defVSLCIFP} with $B$ sorted in decreasing order}
\KwResult{$B_{\min} = (B_1, B_2, \ldots, B_k)$ whose cost is as minimal as possible; cost $e_{\min}$ of this packing}
$e\leftarrow 0$;
divide $L$ into  $T_1$ and $T^+_1$\;
\ForEach{$t$ in $T^+_1$}{ \nllabel{feasibilityAlgoOFFLoopBegin}
\Repeat{$t - b_1 \leq b_1$}{
split $t$ into a fragment $b_1$ and the remainder $t - b_1$\;
pack the fragment $b_1$ into a bin of capacity $b_1$\;
$e\leftarrow e +$ cost of filling up a bin $b_1$
} 
put the remainder $t - b_1$ into $T^-_1$ 
}\nllabel{feasibilityAlgoOFFLoopEnd}
pack items from $ T_1 \cup T^-_1$ to bins $b_1$ with any BPP algorithm\;
$e\leftarrow e +$ cost of this packing ;
$e_{\min} \leftarrow e$; $B_{\min} \leftarrow $ a current packing\; 
\For{$j\leftarrow 2$ \KwTo $m-1$}{
take out items from the less filled bin of size $b_{j-1}$ of packing $B_{\min}$\;
divide them into two categories $T_j$ and $T^+_j$\;
   \ForEach{$t$ in $T^+_j$ which has not yet reached the limit of cuts $D$}{
   split $t$ into 
   $b_j$ and  $t - b_j$;
   pack fragment $b_j$ into a bin $b_j$\; 
   $e\leftarrow e +$ cost of filling up a bin $b_j$;
   put 
   $t - b_j$ into $T^-_j$
   }
pack items from $ T_j \cup T^-_j$ to bins $b_j$ with any BPP algorithm\;
$e\leftarrow e +$ cost of this packing\;
\eIf{$e<e_{\min}$}{
$e_{\min} \leftarrow e$ ; $B_{\min} \leftarrow $ a current packing 
}{
take $B_{\min}$ for repacking
}
}
\ForEach{$B_j$ of $B_{\min}$ taken in decreasing order} 
{ \nllabel{repackingAlgoOFFLoopBegin} 
   \If{$B_j$ is not full $\wedge$ its content may enter into bins of certain capacities $b_l$}
   { 
     find $b_{l_{\min}}$, the smallest of these $b_l$\;
     repack the contents of $B_j$ into an empty bin of capacity $b_{l_{\min}}$ 
   } 
} \nllabel{repackingAlgoOFFLoopEnd} 
\caption{CIFFD solving the VS-CIF-P}
\label{greedyForVSLCIFP}
\end{algorithm}

For any item CIFFD looks for an appropriate opened bin. If it does not find one, it will open up the smallest bin into which the item enters. Its complexity is thus $\mathcal{O}(mn\log n)$.

\begin{theorem}
The VS-CIF-P is $2$-approximable with CIFFD. 
\label{CIFFD-approxTheorem}
\end{theorem}   
\begin{proof}
CIFFD is based upon the consecutive executions of FFD and possibly improving their result due to successful repacking. Taking advantage of Theorem~\ref{theoremApproxCFF} and the fact that the cost is monotone (i.e. a cost of packing is not less than the sum of items to be packed) we obtain also $2$-approximation for CIFFD.
\end{proof}

   \subsection{On-line Approach to the Main Problems with Linear Cost \label{algorithmWithPerformanceBounds}}
The algorithmic on-line method we propose now is founded upon FFf (see Paragraph~\ref{FFf}) with item cuts incorporated (CFFf).
As in Subsection~\ref{greedyApproach} and for the same reason, we operate on item sizes, not on item indices. To keep the notation brief, we put $b_{\max} = \max_{b_i}(B)$.

In a nutshell, the CFFf idea is as follows. For items whose sizes are smaller than the largest bin capacity Hypothesis~\ref{StrongH} holds. These items, which form set $A$, can be therefore packed optimally with NFC (Subsection~\ref{NFC}). Other items, which constitute set $A^+=L-A$, require a split before packing.
For any element $t$ of $A$ we perform a cut into $b_{\max}$ and $t-b_{\max}$ fragments. The remainders $t-b_{\max}$ form set $A^-$ and they are packed according to FFf with $f$ indicating a fill factor (Paragraph~\ref{FFf}).

We believe that this explanation is sufficient to implement the algorithm.
We propose, however, in Algorithm~\ref{approxForVSLCIFP}, a more detailed description which shows explicitly a classification of items from set $A^-$ (i.e. items which are the remainders of cuts) into categories which are induced by different manners of item packing.
These three categories of items from set $A^-$, which we enumerate and comment on below, play an important role in the approximability proof:
\begin{itemize}
\item $X$ --- items packed individually into bins of capacity $b_{\max}$,
\item $Y$ --- items packed into bins of capacity $b_{\max}$ sharing them with other items,
\item $Z$ --- items packed into bins of any capacity $b$, $b<b_{\max}$.
\end{itemize}

\begin{algorithm}[!h] 
\KwData{VS-CIF-P data as in Def.~\ref{defVSLCIFP}, a fill factor $f$}
\KwResult{a packing whose cost is as minimal as possible}

\ForEach{$t$ in $L$}{ 

\eIf(\tcc*[f]{$t$ from $A$}){$t\leq b_{\max}$}
{ pack $t$ into bins of capacity $b_{\max}$ with NFC\;\nllabel{NFCinCFFf}
continue
}(\tcc*[f]{$t$ from $A^+$})
{\Repeat{$t\leq b_{\max}$}{split $t$ into a fragment $b_{\max}$ and the remainder $t - b_{\max}$\;\nllabel{beginRepeatCFFf}
pack the fragment $b_{\max}$ into a bin $b_{\max}$\;
$t \leftarrow t - b_{\max}$
} \nllabel{endRepeatCFFf} 
}
\tcc*[h]{$t$ is from $A^-$}

\If{there is room for $t$ in an opened bin}{\nllabel{beginIntelligentPackCFFf} 
                     $b \leftarrow $ the first opened bin into which $t$ can be packed\;          
 		      pack $t$ into b
 		\tcc*[f]{$t\in Y$ when the bin $b_{\max}$ is used and $t\in Z$ otherwise}
		}
\tcc*[h]{an empty bin has to be opened for $t$}

\eIf(\tcc*[f]{$t$ is small}){$t\leq 0.5\cdot b_{\max}$}{ 
		{
  		pack $t$ into an empty bin of capacity $b_{\max}$ \tcc*[f]{$t\in Y$}
 		} 
 }(\tcc*[f]{$t$ is big})
 {
 	  \eIf{there are bins of capacity between $t$ and $\frac{t}{f}$}
	      {
		$b \leftarrow $ the smallest empty bin of capacity between $t$ and $\frac{t}{f}$ 
            }
             {
               $b \leftarrow$  a bin of capacity $b_{\max}$
             }
          pack $t$ into bin $b$
          \tcc*[f]{$t\in Z$ if $b<b_{\max}$ and $t\in X$ or  $t\in Y$ if $b=b_{\max}$}
       }\nllabel{endIntelligentPackCFFf} 
} 
\caption{CFFf solving the VS-CIF-P}
\label{approxForVSLCIFP}
\end{algorithm}

The computational effort of CFFf is concentrated upon searching an appropriate bin among those which have been already opened and selecting an empty bin with respect to a given fill factor $f$. For $f=0.5$ the complexity of CFFf is $\mathcal{O}(n(\log n + \log m) +m\log m)$.

We estimate the quality of solution obtained with CFFf for an instance $I$ with list $L$ of items to be packed, $\mbox{CFFf}(L)$. We assume here that the cost is linear and, moreover, a bin cost is equal to its capacity, $b_i=c_i$, $i=1,2,\ldots, m$ as stated in Subsection~\ref{VSBPCP}.  Any solution cost is always less than or equal to the overall mass of items from $L$: $S_L=\sum_{t\in L}t$. The notation $S_C$ indicates later on a sum of item sizes from any set $C$.
\begin{theorem}
$\mbox{CFFf}(L)\leq \frac{4}{3}S_L + 2b_{\max} $.
\label{CFFf-approxTheorem}
\end{theorem}
\begin{proof}
 As NFC, which packs items from $A$ is exact and polynomial (Theorem~\ref{theoremNFC}), $\mbox{CFFf}(A)\leq S_A + b_{\max}$.
We have to estimate the packing quality for items from $A^-$ which are divided into three categories: $X$, $Y$, and $Z$ (see Algorithm~\ref{approxForVSLCIFP}).

Obviously, as items of category $X$ result from splitting and they occupy bins singly, $\mbox{CFFf}(X) = 2|X|b_{\max}$ and $S_X\geq 1.5 |X|b_{\max}$ with exception to at most a single bin, which gives:
\begin{equation}
\mbox{CFFf}(X) \leq \frac{4}{3}S_X + b_{\max}.
\label{approxX}
\end{equation}
Let $Y_B$ stand for these items of Y which are packed into bin B. Analogously, $\mbox{CFFf}(Y_B) = (|Y_B|+1)b_{\max}$ and $S_Y \geq (|Y_B|+\frac{2}{3})b_{\max}$ with exception to at most a single bin. This leads to:
\begin{equation}
\mbox{CFFf}(Y) \leq \frac{9}{8}S_Y + b_{\max}.
\label{approxY}
\end{equation}
For a bin of capacity $b$ in which items $Z_b$ of $Z$, $Z_b\subset Z$, are packed we have $\mbox{CFFf}(Z_b) = |Z_b| b_{\max} + b$ and $S_{Z_b} > b_{\max} +fb$. Consequently,
\begin{equation}
\mbox{CFFf}(Z) \leq \frac{3}{2+f}S_Z.
\label{approxZ}
\end{equation}

Combining the inequalities~\eqref{approxX}--\eqref{approxZ} with $f=0.5$ we get
\[\mbox{CFFf}(A^-) \leq \frac{4}{3}S_{A^-} + 2b_{\max}\]
which proves the theorem as $ \mbox{CFFf}(L) = \mbox{CFFf}(A) + \mbox{CFFf}(A^-)$ as the number of completely filled bins $b_{\max}$ has already been counted. 
\end{proof}

\section{Performance Evaluation \label{numericalResults}}

The goal of the performance analysis is to estimate the difference of results obtained with our approximation algorithms relative to the exact solutions. Moreover, we oppose our methods to two simple reference
algorithms, less ``intelligent'' and less costly in terms of computational effort. We also analyze the impact of the number of bin classes available and the number of cuts allowed $D$ on the approximation ratio of the obtained results.

\subsection{Experimental Setup \label{experimentalSetup}}

In order to estimate algorithm's approximation ratios we created VS-CIF-P instances from exact solutions artificially made. We also proceeded with the comparaison of algorithms' results for instances whose exact solutions are unknown. 

All instances treated are feasible, i.e. the largest item fragmented at most $D$ times can be inserted into bins of the greatest capacity. The numerical experiments were conducted for instances with few  ($m=3$) and many bin classes ($m=10$). We arbitrarily fixed the largest capacity $b_{\max}$ to $100$. The capacities of other $m-1$ classes are chosen uniformly in the natural interval $[1,b_{\max}-1]$.

In the case of a linear cost we took a bin cost equal to its capacity, $c_i = b_i$, $i=1\ldots, m$. When dealing with a monotone cost we assume that the largest bin cost is also equal to its capacity, $c_{\max} = b_{\max}$. Assuming that bin classes are sorted in decreasing order of their size, the cost $c_{i+1}$ is chosen uniformly in $[b_{i+1},c_i -1]$. 

Initial items are generated uniformly in $[1,99]$ and their average size is $50$. The number of initial items is arbitrarily fixed to $200$. We have thus the expected total volume of $10'000$ to be packed.

The construction of exact solutions consists in putting initial items into bins with FF and filling the opened bins entirely with extra items. Items to feed the algorithms which admit fragmentation are made up of initial and extra items put together: up to $D+1$ items (initial or extra) can be glued to form an individual item. Consequently, in the experiments we made, for $D=1$ the average item size is equal to $100$, for $D=2$ the average item size is equal to $200$, etc.

The exact solutions which we took as a base of the instance creation are composed typically of numerous large bins and a single little one. Not willing to restrict ourselves to such a solution form we also realized the direct confrontation of results obtained from initial items, possibly glued up to $D$ times, as explained above. We did not add, however, any extra items as we did not fill up the opened bins. By giving up fixing an exact solution as a starting point of the instance construction we award instances with more flexibility.   

Despite the fact that the approximation factor we give for CFFs in Theorem~\ref{CFFf-approxTheorem} holds for the linear cost, we decided to run CFFf with the monotone cost, too. We argue that at this stage we can evaluate empirically its performance with the monotone cost.

The results are averaged for series of $1000$ instances with which the algorithms are fed. The confidence intervals depicted in all figures illustrating the following subsection are computed with the confidence level $\alpha=0.05$. 

\subsection{Results \label{results}}

Before presenting the results we explain the simple greedy algorithms which will be brought face to face with our algorithms.

The first of them, called CNFL (this abbreviation is straightforward and will be explained below) is on-line. Its operating mode is two-fold. First, it cuts items up to $D$ times to fit their fragments into largest bins. This is a ``modulo $b_{\max}$ cut'': $D$ bins are filled up, the last one may be partially filled. Next, it packs them according to the Next Fit principle (Cut and Next Fit, CNF). The reader may refer to Subsection~\ref{lessConstrantedProblemDefinition} and Figure~\ref{exampleCFF} to recall the discussion of the similar CFF based upon ``almost equal'' cuts. Cost minimizing is obtained by always using the Largest bin (the dual principle to the one seen in Paragraph~\ref{FFf}). Assuming that bin classes are preliminarily sorted, the CNFL complexity is $\mathcal{O}(n)$. Its performance will be compared with that of CFFf.
\begin{figure}
  \begin{center}
    \includegraphics[width=0.7\columnwidth]{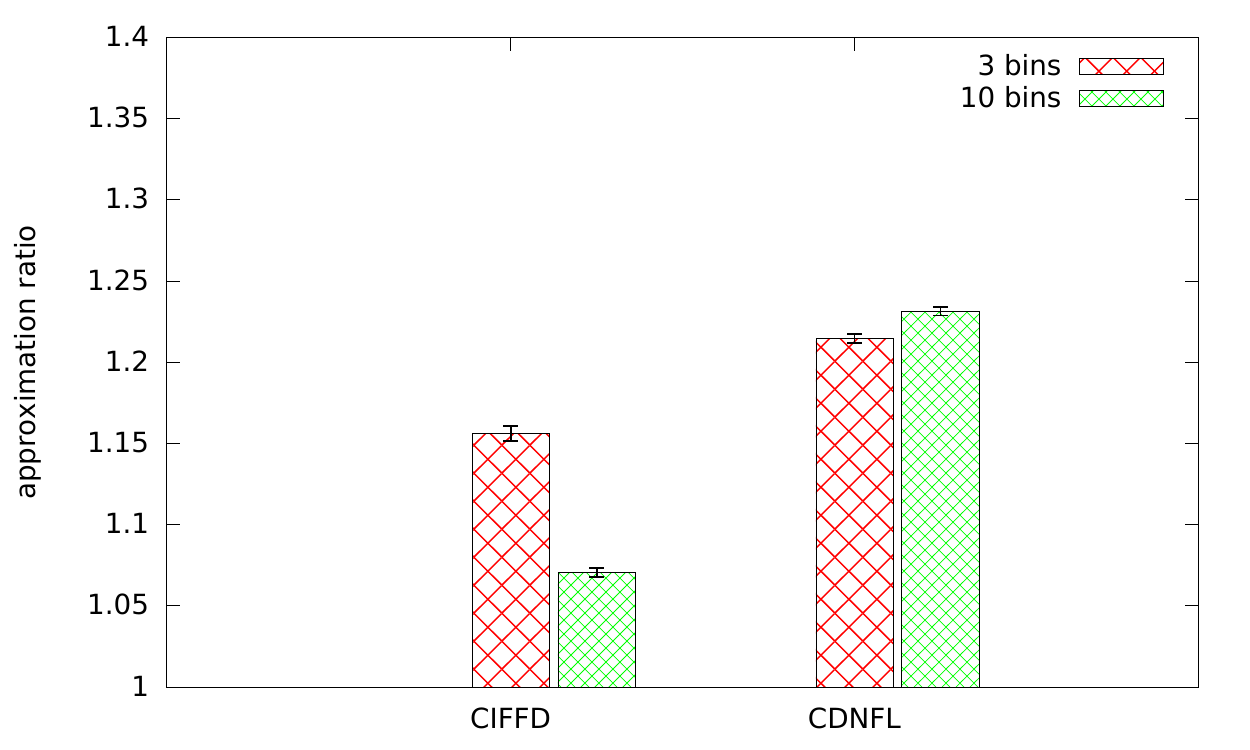}
  \end{center}
  \caption{\label{compOFFGloutonLinearFunctionBinClasses}Estimation of the approximation ratio for the off-line algorithms, CIFFD and CDNFL, with $D=1$ and linear cost for different numbers of bin classes}
\end{figure}
\begin{figure}[h]
  \begin{center}
    \includegraphics[width=0.7\columnwidth]{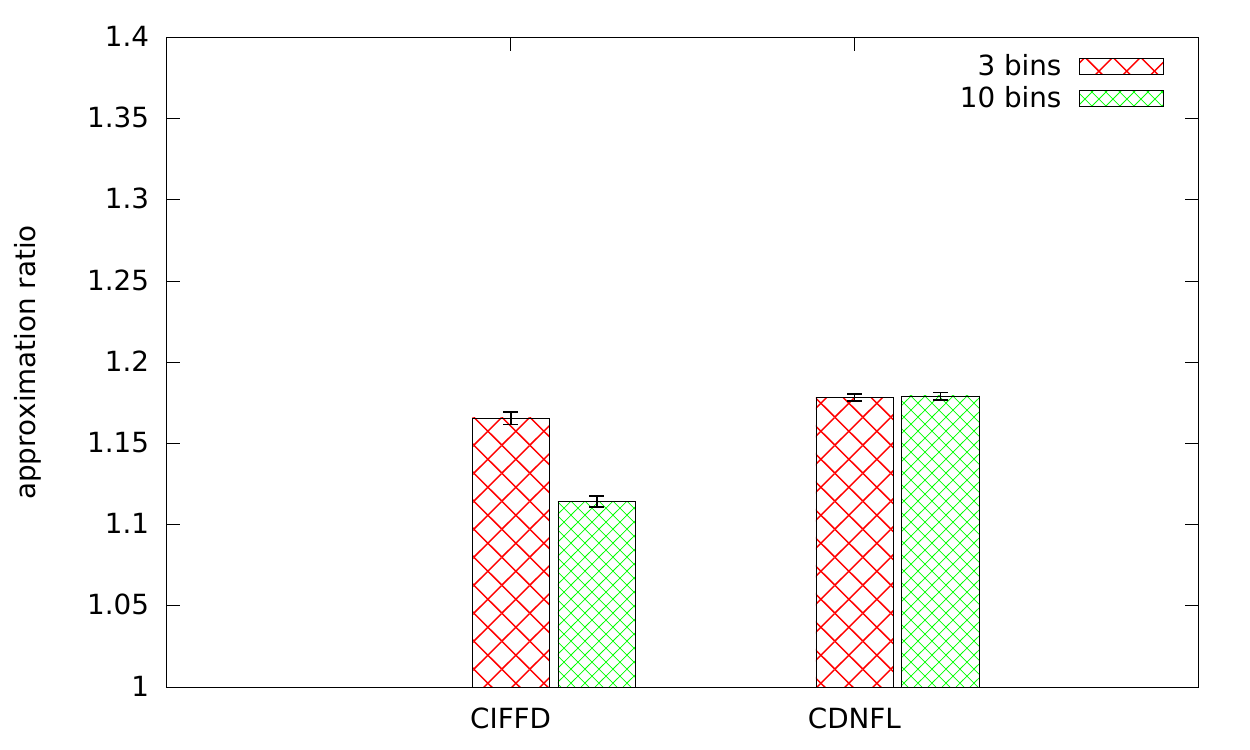}
  \end{center}
  \caption{\label{compOFFGloutonMonotoneFunctionBinClasses}Estimation of the approximation ratio for the off-line algorithms, CIFFD and CDNFL, with $D=1$ and monotone cost for different numbers of bin classes}
\end{figure}

The second one is an off-line mutation of CNFL in which the items, after the preliminary ``modulo $b_{\max}$ cut'' made as explained above, are sorted in decreasing order. This off-line algorithm, to be confronted with CIFFD, is obviously called CDNFL.

Intuitivelly, the great number of cuts allowed may facilitate packing procedures. We opted thus to confront the algorithms for  $D=1$.
Figures \ref{compOFFGloutonLinearFunctionBinClasses} and \ref{compOFFGloutonMonotoneFunctionBinClasses} present the results of the comparison of the approximation ratios obtained with two off-line algorithms in function of the number of bin classes for linear and monotone costs, respectively. Figures \ref{compONGloutonLinearFunctionBinClasses} and \ref{compONGloutonMonotoneFunctionBinClasses} do the same for both on-line methods.

Figures \ref{compOFFGloutonMonotoneFunctionBinClasses} and \ref{compONGloutonLinearFunctionBinClasses} show at a glance that the algorithms perform much better the theoretical performance bounds given in Theorems~\ref{CIFFD-approxTheorem} and \ref{CFFf-approxTheorem} for CIFFD with monotone cost and CFFf with linear cost, respectively. 

As one may expect, the approximation ratio obtained with CIFFD is significantly better comparing with the one produced by the naive approach for both the analyzed costs (Figures~\ref{compOFFGloutonLinearFunctionBinClasses} and \ref{compOFFGloutonMonotoneFunctionBinClasses}). As our algorithm is based upon consecutive repacking of a single, the least filled bin, the impact of the number of bin classes available is considerable. CIDDF packs better when having many bin classes at its disposal, in contrast to CDNFL which is insensitive to this parameter.  

The on-line approach, CFFf, is not significantly influenced by the number of bin classes. Figures~\ref{compOFFGloutonLinearFunctionBinClasses}--\ref{compONGloutonMonotoneFunctionBinClasses} put in evidence its strikingly good performance. CFFf, despite being on-line, outperforms even the off-line CIFFD method in certain situations. This interesting phenomenon will be explained below while studying the influence of the number of cuts allowed.
\begin{figure}
  \begin{center}
    \includegraphics[width=0.7\columnwidth]{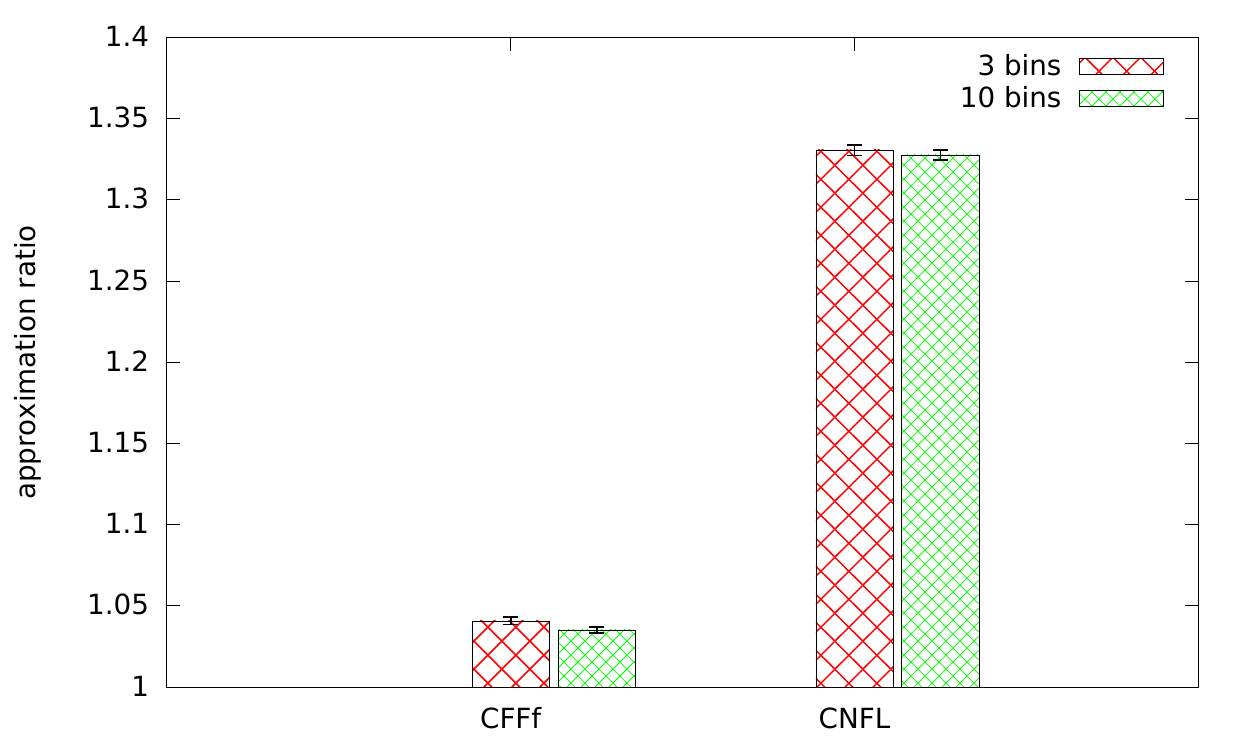}
  \end{center}
  \caption{\label{compONGloutonLinearFunctionBinClasses}Estimation of the approximation ratio for the on-line algorithms, CFFf and CNFL, with $D=1$ and linear cost for different numbers of bin classes}
\end{figure}
\begin{figure}
  \begin{center}
    \includegraphics[width=0.7\columnwidth]{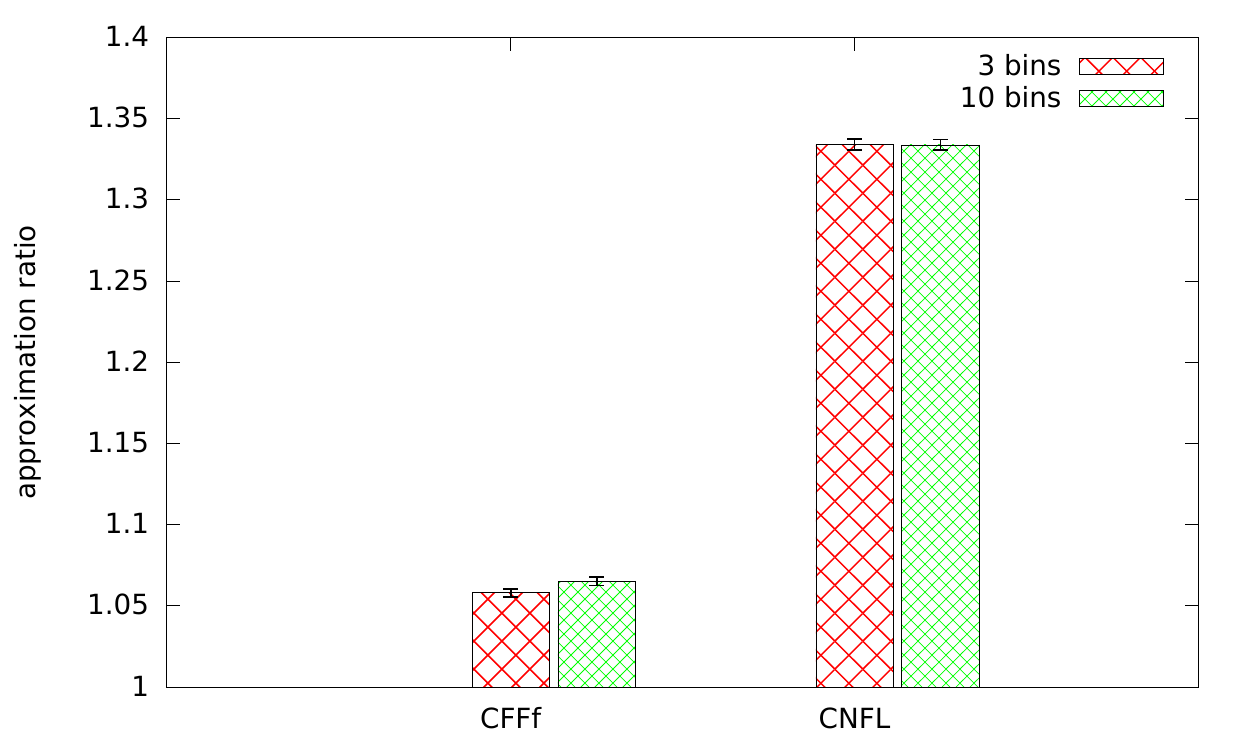}
  \end{center}
  \caption{\label{compONGloutonMonotoneFunctionBinClasses}Estimation of the approximation ratio for the on-line algorithms, CFFf and CNFL, with $D=1$ and monotone cost for different numbers of bin classes}
\end{figure}

Figures \ref{relativeLinearFunctionBinClasses} and \ref{relativeMonotoneFunctionBinClasses} show the performance of the algorithms for the same series of instances whose exact solutions are {\em a priori\/} unknown, with linear and monotone cost, respectively. This experience allowed us to compare the algorithm quality for instances which do not suffer from the imposed form of an exact solution. The smaller value of the average packing cost signifies a higher packing efficiency. As the average total volume of items to be packed is preserved and equal to $10'000$, the reader may observe the similar tendency as in the case of the comparison with optimal solutions. It is not astonishing that all algorithms behave better for the linear cost. CIFFD becomes more efficient when the number of bin classes goes up. Again, CFFf performs much better that a greedy off-line method CNFL.

The impact of the limit set on the number of splits permitted is illustrated in Figures \ref{impactD_OFF_monotoneForDifferentBinClasses} and \ref{impactD_ON_monotoneForDifferentBinClasses} for the off-line and on-line approaches, respectively. This analysis reveals a secret of the excellent performance of CFFf. As the results for the linear and monotone costs exhibit the same tendency, we restrict the graphical presentation to the latter only. The fragmentation ban ($D=0$) signifies that the problem solved is simply the VSBPCP.

The graph in Figure~\ref{impactD_OFF_monotoneForDifferentBinClasses} confirms the intuition that more splits allowed make packing easier. For instance, CIFFD with monotone cost, $10$ bin classes and up to $8$ cuts often reaches ``an almost exact solution''.

The behavior of CFFf depicted in Figure~\ref{impactD_ON_monotoneForDifferentBinClasses} does not, however, exhibit the same trend. Before explaining this phenomenon we recall to the reader that the items which satisfy Hypothesis~\eqref{StrongH} are packed optimally with NFC according to Theorem~\ref{theoremNFC}. We also call up that in our experiments the average total volume of items to be packed is preserved regardless the value of $D$ (see Subsection~\ref{experimentalSetup}). It means that for a great $D$ value an instance has less items but they are bigger. 

As we see in Figure~\ref{impactD_ON_monotoneForDifferentBinClasses}, admitting one cut ($D=1$) drastically lowers the solution cost comparing with the situation when the fragmentation is forbidden (the VSBPCP for $D=0$). In our experiments, the average item size for $D=1$ is $100$. The largest bin has the same capacity $b_{\max}=100$. A relatively large part of items is therefore inserted into largest bins by NFC (line~\ref{NFCinCFFf} in Algorithm~\ref{approxForVSLCIFP}). When more cuts are allowed, for example $D=2$, the average item size is greater, $200$, while the largest bin capacity stays unchanged, $b_{\max}=100$, the ``modulo $b_{\max}$ splitting'' made in the repeat loop (lines~\ref{beginRepeatCFFf}--\ref{endRepeatCFFf} of Algorithm~\ref{approxForVSLCIFP}) takes over. This loop may potentially open up too many largest bins than necessary. Finally, when $D$ is increasing (starting from $D=7$ in our experiments) the negative impact of the repeat loop is 
compensated by the intelligent packing realized in lines~\ref{beginIntelligentPackCFFf}--\ref{endIntelligentPackCFFf} of Algorithm~\ref{approxForVSLCIFP} and the CFFf performance stabilizes close to an exact solution (between five and ten per cent).

We thus draw a conclusion that the more items CFFf inserts with NFC, the better its performance is. 
\begin{figure}
  \begin{center}
    \includegraphics[width=0.7\columnwidth]{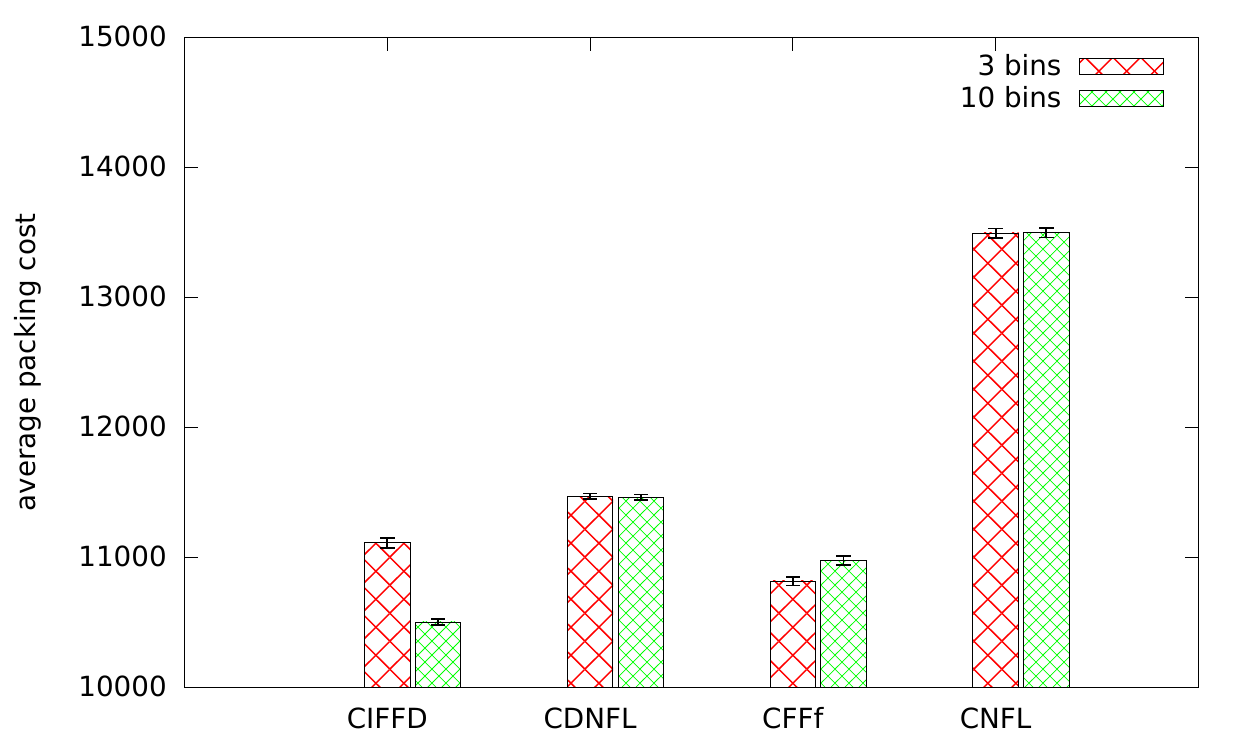}
  \end{center}
  \caption{\label{relativeLinearFunctionBinClasses}Average result of the off-line (CIFFD and CDNFL) and on-line algorithms (CFFf and CNFL) with $D=1$ and linear cost for different numbers of bin classes}
\end{figure}
\begin{figure}
  \begin{center}
    \includegraphics[width=0.7\columnwidth]{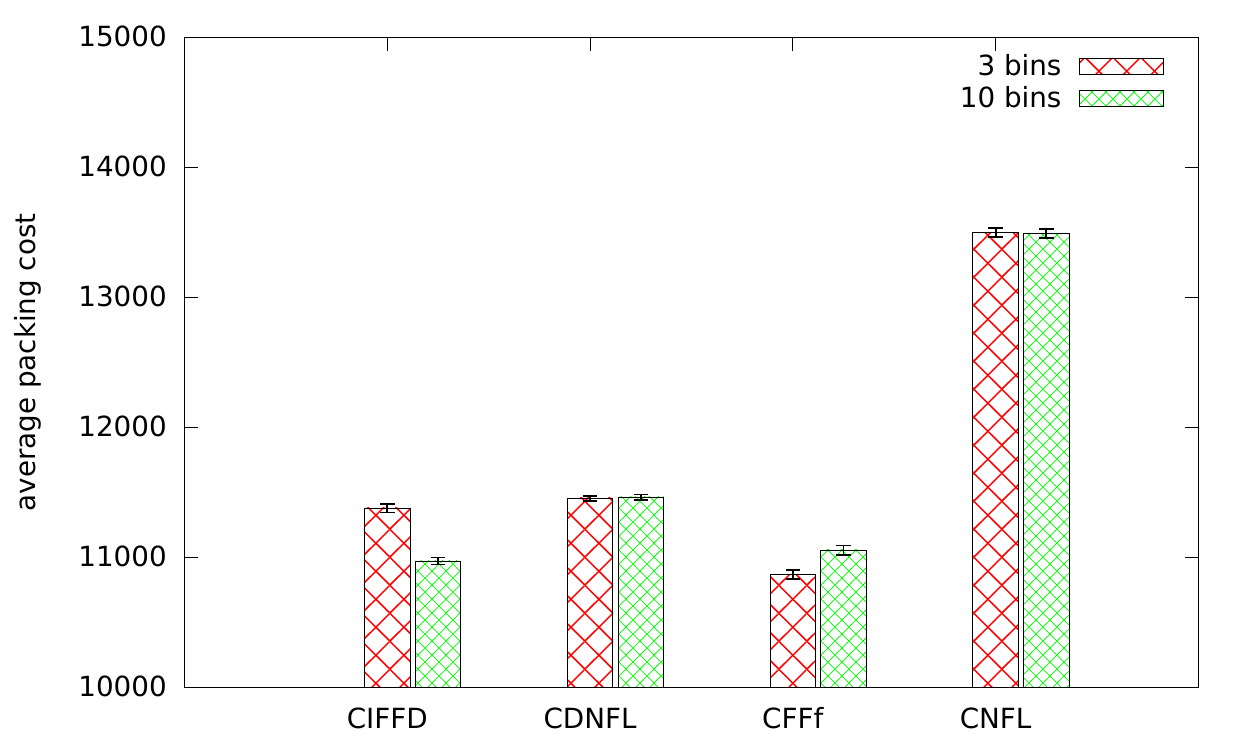}
  \end{center}
  \caption{\label{relativeMonotoneFunctionBinClasses}Average result of the off-line (CIFFD and CDNFL) and on-line algorithms (CFFf and CNFL) with $D=1$ and monotone cost for different numbers of bin classes}
\end{figure} 
\begin{figure}
  \begin{center}
    \includegraphics[width=0.7\columnwidth]{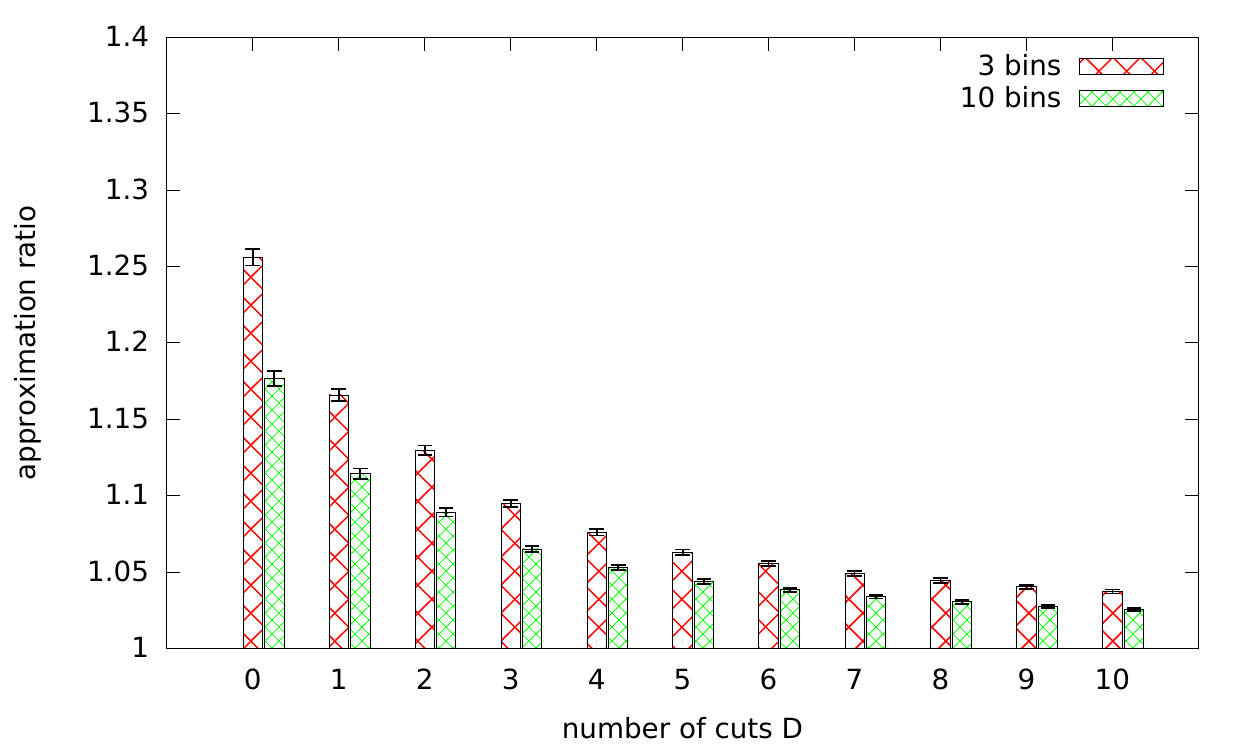}
  \end{center}
  \caption{\label{impactD_OFF_monotoneForDifferentBinClasses}Estimation of the approximation ratio for the off-line CIFFD with monotone cost for different numbers of bin classes available in function of the number of splits allowed}
\end{figure}
\begin{figure}
  \begin{center}
    \includegraphics[width=0.7\columnwidth]{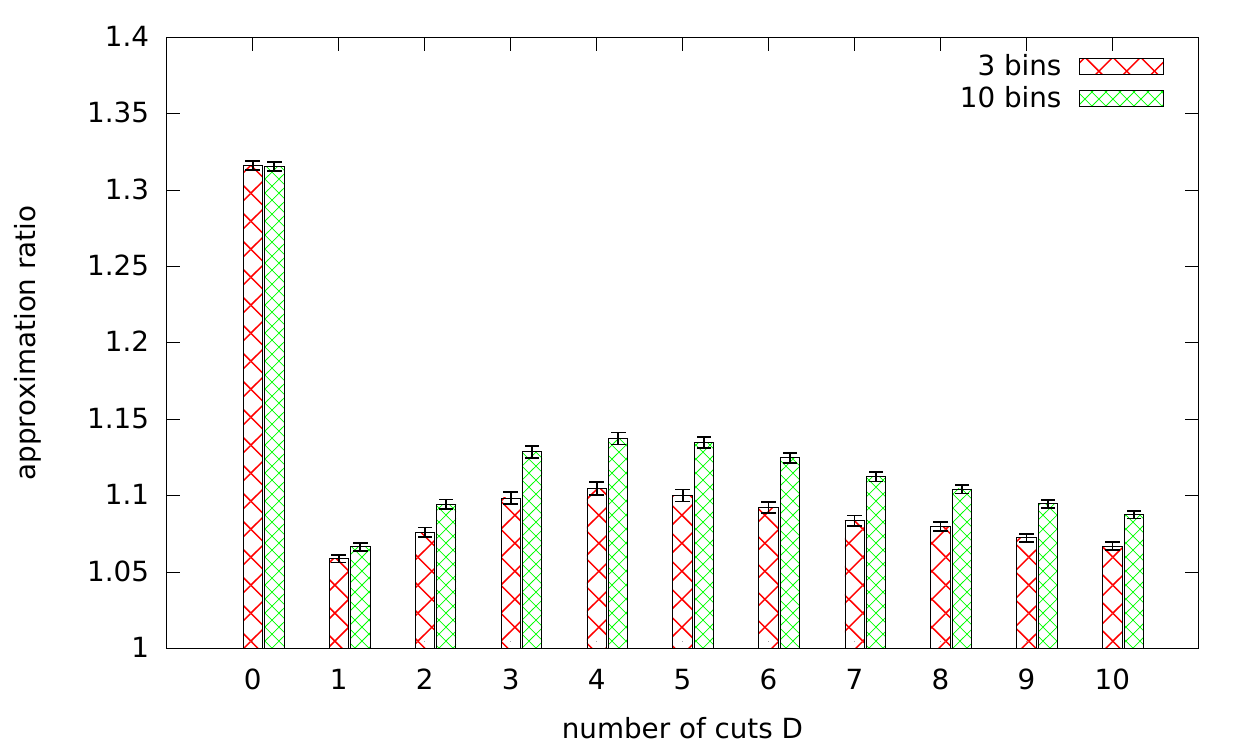}
  \end{center}
  \caption{\label{impactD_ON_monotoneForDifferentBinClasses}Estimation of the approximation ratio for the on-line CFFf with monotone cost for different numbers of bin classes available in function of the number of splits allowed}
\end{figure}

\section{Conclusions and Perspectives \label{conclusionsAndPerspectives}}
We proposed the modeling of the energy-aware load balancing of computing servers in networks providing virtual services by a generalization of the Bin Packing problem. As a member of the Bin Packing family, our problem is also approximable with a constant factor. In addition to its theoretical analysis we proposed two algorithms, one off-line and another one on-line, giving their theoretical performance bounds. The empirical performance evaluation we realized showed that the results they provide are significantly below the approximation factor. 

On the one hand, this practical observation encourages us to continue looking for a better approximation factor, especially for CIFFD whose performance is much better than predicted theoretically. We will also try to extend the analysis of CFFf to the monotone cost.

On the other hand, we gave the performance evaluation which allows us to consider CIFFD and CFFf as the algorithms with a very good potential for practical applications like energy-aware load balancing, which motivated our work. We emphasize the remarkable efficiency of our on-line approach (CFFf) which outperforms considerably simple off-line algorithms. Notwithstanding, the theoretical result proven only for the linear cost CFFf behaves very well when bin costs are monotone.

The approach presented in this paper is centralized and adapted to private Cloud infrastructures. Another challenge calling into question is the application of our packing approach into a distributed environment when information about resource availability is incomplete. 

We believe that the approach through packing is a powerful tool allowing one to perform a load-balancing in Clouds which ensures the realization of tasks with respect to their requirements while consuming the smallest quantity of electrical energy. For this reason we think to use this approach in our further multi-criteria optimization of a Cloud infrastructure. Among other criteria which we find  essential to study in this context are the efficient utilization of resources of a telecommunication network (the principle ``network-aware Clouds'') and the guarantee of meeting the QoS requirements expressed in customers' contracts, concerning, for instance, the task termination before a given dead-line or the task execution time limited by a given make-span.   

\section*{Acknowledgment}
St\'ep{}hane Henriot participation in this work was partially financed by the grant 2013-22 of PRES UniverSud Paris.

\bibliographystyle{IEEEtran}
\bibliography{suscomForCorr}

\end{document}